\documentclass[a4paper,10pt]{article}
\usepackage[utf8]{inputenc}
\usepackage[english]{babel}
\usepackage[left=3cm,right=3cm,bottom=3cm,top=3cm]{geometry}
\usepackage{graphicx}
\usepackage{amsthm}
\usepackage{amsfonts}
\usepackage{amsmath}
\usepackage{amssymb}

\newtheorem{theorem}{Theorem}
\newtheorem{Lemma}{Lemma}
\newtheorem{Proposition}{Proposition}
 \theoremstyle{definition}
\newtheorem{Definition}{Definition}

\begin{document}

\title{Extending de Bruijn sequences to larger alphabets}

\author{\begin{tabular}{cc}
Ver\'onica Becher \hspace{1cm}& \hspace{1cm}Lucas Cort\'es\\
{\small vbecher@dc.uba.ar}\hspace{1cm} &\hspace{1cm} {\small lucascortes@me.com}
\end{tabular}\\
\small Departamento de  Computaci\'on,
Facultad de Ciencias Exactas y Naturales \& ICC \\
\small Universidad de Buenos Aires \&  CONICET, Argentina 
}

\maketitle

\begin{abstract}
A  de Bruijn sequence of order $n$ over a  $k$-symbol alphabet 
 is a circular sequence where each length-$n$ sequence  occurs exactly once. 
We present a way of extending  de Bruijn sequences by 
adding a new symbol to the alphabet: 
the extension is performed by embedding a given de Bruijn  
sequence into another one of the same order, but over the alphabet with one more symbol, 
while  ensuring that there are no long runs without the new symbol.  
Our solution  is based on auxiliary graphs derived from the de Bruijn graph
and solving a problem of maximum~flow.
\end{abstract}

\noindent
{\bf Keywords}:  de Bruijn sequences,   Eulerian cycle, maximum flow, combinatorics on words.

{\small \tableofcontents}

\section{Introduction and statement of results}

A circular sequence is the equivalence class of a sequence under rotations. 
A  de Bruijn sequence of order $n$  over a $k$-symbol alphabet 
is a circular sequence of length $k^n$  in which every  length-$n$ sequence 
 occurs exactly once~\cite{db46,saintemarie},  see~\cite{berstel} for a fine presentation and history. 
For example,  writing  $[abc]$ to denote the circular sequence formed by the rotations of $abc$, 
$ [0011]$ is de Bruijn of order~$2$ over the alphabet $\{0,1\}$.

A  subsequence of a sequence $a_1 a_2 \ldots a_n$ is a sequence 
$b_1 b_2 \ldots b_m$ defined by $b_i = a_{n_i}$ for $i=1, 2, \ldots, m$, where $n_1 \leq n_2 \leq \ldots \leq n_m$. 
 The same applies to circular sequences, assuming any starting position. 
For example, for the alphabet of digits from $0$ to $9$, 
  $[123]$, $[246]$ and [5612] are subsequences of $[123456]$.

Clearly,  for any given de Bruijn sequence over a $k$-symbol  alphabet 
there is another one  over  the  alphabet enlarged with one  new symbol, 
 such that the two sequences have the same order, 
and the first  is a subsequence of  the second. 
This is immediate from the characterization of de Bruijn sequences 
as Eulerian cycles on de Bruijn graphs: the de Bruijn graph 
for the original alphabet is a sugbgraph of the de Bruijn graph 
for the enlarged alphabet, 
and  any cycle in an Eulerian graph can be embedded into a full Eulerian cycle. 
For instance, such an extension can be constructed  with  Hierholzer's algorithm for joining 
cycles together to create an Eulerian cycle of a graph.
However, this gives no guarantee that  the new symbol  
is fairly distributed along the resulting de Bruijn sequence.

In this note we consider the problem of extending a de Bruijn sequence 
over a $k$-symbol alphabet  to another one  of the same order  over 
the alphabet enlarged with a new symbol,
such that the first is a subsequence of the second 
  {\em and there are no long runs without the new symbol}.
If  in between every two successive occurrences of the new symbol there were  fewer than $n$ symbols, 
it would be impossible
 to  accommodate all words of length $n$ lacking  the new symbol.
If  there were exactly~$n$ symbols, 
 to accommodate all words of length~$n $ lacking the new symbol
we would need~$(n+1) k^n$ symbols. 
But this would be  impossible because for all sufficiently large values of $n$  this quantity exceeds~$(k+1)^n$,
the length of a de Bruijn sequence of order~$n$ over a $(k+1)$-symbol alphabet. 
Theorem~\ref{thm:1} proves that there is an extension that  in between 
any two successive occurrences of the new symbol there can be 
at most~$n+2k-2$ other symbols.

\begin{theorem}   \label{thm:1}
For any  de Bruijn sequence $v$ over a $k$-symbol alphabet of order~$n$  there is another one~$w$
 over  that  alphabet enlarged with a new symbol,  of the same order $n$, 
such that~$v$ is a subsequence of~$w$  and for any $n + 2k-1$ consecutive symbols in~$w$ 
there is at least one occurrence of the new symbol.
\end{theorem}
For example, for this  de Bruijn  sequence of order~$3$  over  the alphabet $\{0,1\}$,
\[v = [11000101]
\]
the following  de Bruijn sequence of order~$3$  over the alphabet $\{0,1,2\}$
 satisfies the conditions of the theorem: 
\[
w = [122212111002202000120102101]
\]
because $v$ is a subsequence of $w$ and given any $n+2k-1 = 6$  consecutive symbols in $w$ 
there is at least one occurrence of the symbol~$2$.

To prove Theorem \ref{thm:1}, in addition to  classical elements from graph theory such as de Bruijn graphs, 
Eulerian cycles and graph transformations, we pose the  fairness condition   on the new symbol
as a  problem of maximum flow 
and solve it with Edmonds-Karp algorithm~\cite{ek,Cormen}.
The following is a crude  upper bound of the complexity of the construction.

\begin{Proposition}\label{prop:1}
For order $n$ and every $k$-symbol alphabet there  is a construction that proves Theorem \ref{thm:1} 
in  $O(k^{3n-2})$ mathemtical operations.
\end{Proposition}

It is possible to conceive this extension problem 
 in variants  of de Bruijn sequences  defined in terms of Eulerian cycles in approrpriate graphs.
For instance, the {\em semi-perfect} de Bruijn sequences of   Repke and  Rytter~\cite{RepkeRitter2018}
which satisfy that   each of the  prefixes (large enough) has the largest possible number of distinct words.
Or the {\em perfect}  sequences~\cite{ABFY} 
which, for order $n$,  contain each word of length~$n$ exactly~$n$ times 
but each one starting at different positions modulo~$n$.
Or the subtler {\em nested perfect} sequences~\cite{BC2019} originated in Mordachay Levin's~\cite[Theorem 2]{Levin1999}.

The extension to a larger alphabet 
without the fairness condition on the new symbol
is particularly simple for the  lexicographically greatest de Bruijn sequence:
the one over the original alphabet is   the suffix of the one of the enlarged alphabet~\cite{Thibeault},
assuming the new symbo is the lexicographically greatest.
The extension can be done with an efficient greedy algorithm, see~\cite{SSW}.

The extension problem to a larger alphabet that we consider in the present note
is dual to the extension problem studied by Becher and Heiber
 in~\cite{BH2011}, where they considered  extending 
a de Bruijn  sequence of order~$n$ over a $k$-symbol alphabet to 
another  one of order $n+1$ over the same alphabet such that 
the first is a prefix of the second.

\section{Proof of Theorem \ref{thm:1}} \label{sec:db}

In the sequel  we use the terms word and sequence interchangeably.
A de Bruijn graph $G(k, n)$ is a directed graph  whose vertices are 
the words of length  $n$ over a $k$-symbol alphabet 
and whose edges  are the  pairs $(v,w)$  where $v=au$ and $w=ub$, for some word $u$
of length $n-1$ and possibly two different symbols $a, b$.
Thus, the graph $G(k, n)$ has $k^n$ vertices and $k^{n+1}$ edges, 
it is strongly connected and every vertex
 has the same in-degree and out-degree. 
Each  de Bruijn sequence of order $n$ over a $k$-symbol alphabet can be constructed by taking a Hamiltonian 
cycle in~$G(k,n)$.
 Since the line graph of $G(k,n)$ is $G(k,n+1)$, 
each de Bruijn sequence of order $n+1$ over a $k$-symbol alphabet 
can be constructed as an Eulerian cycle in~$G(k,n)$. 

\subsection{Graph of circular words}

\begin{figure}[t!]
  \centering
   \includegraphics[width=0.6\textwidth]{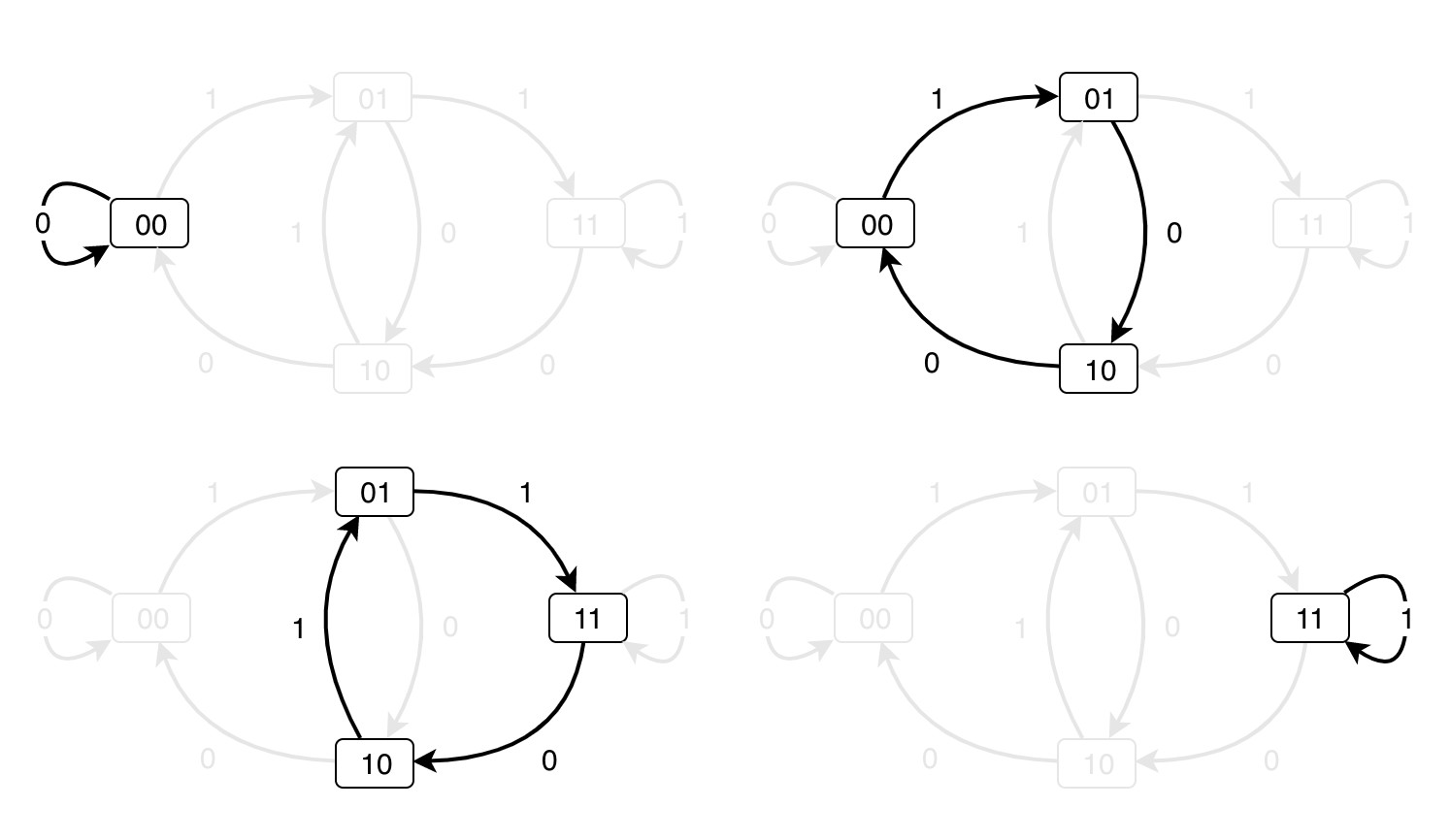}
  \caption{For  alphabet $\{0,1\}$  there are $4$ circular words of length $3$: 
$[000]$, $[100]$, $[110]$ and $[111]$, each corresponds to a simple  cycle in the de Bruijn graph $G(2, 2)$.}
  \label{fig:simple}
  \includegraphics[width=\textwidth]{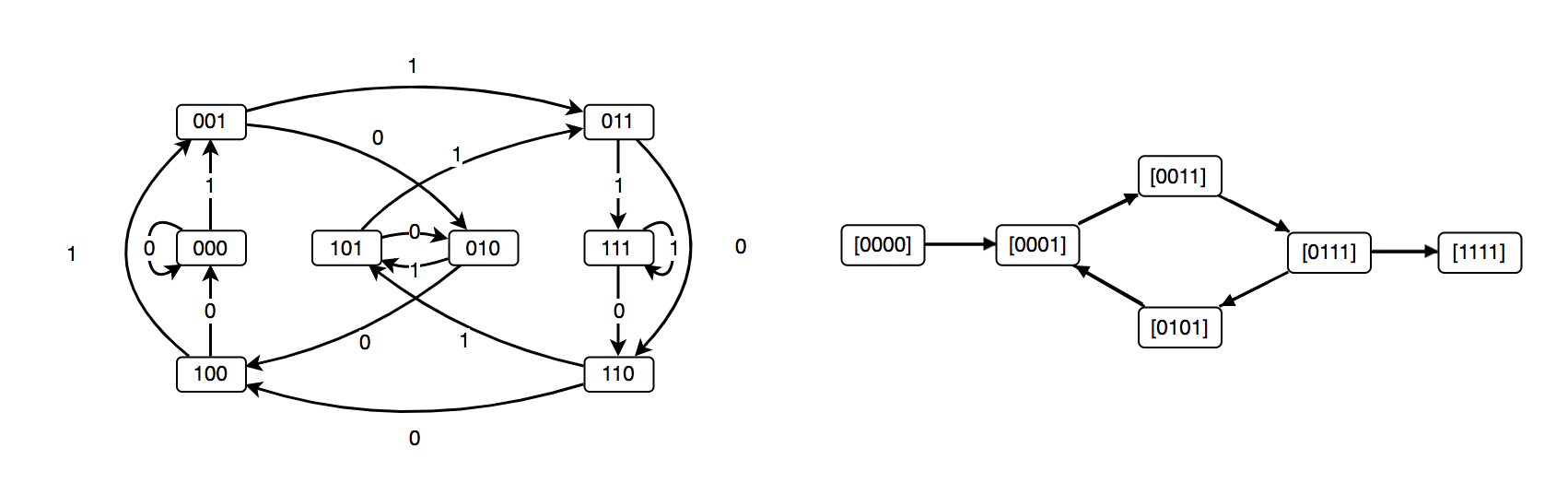}
  \caption{On the left $G(2,3)$. On the right graph $C(2,4)$.}
  \label{fig:C}
\includegraphics[width=0.45\textwidth]{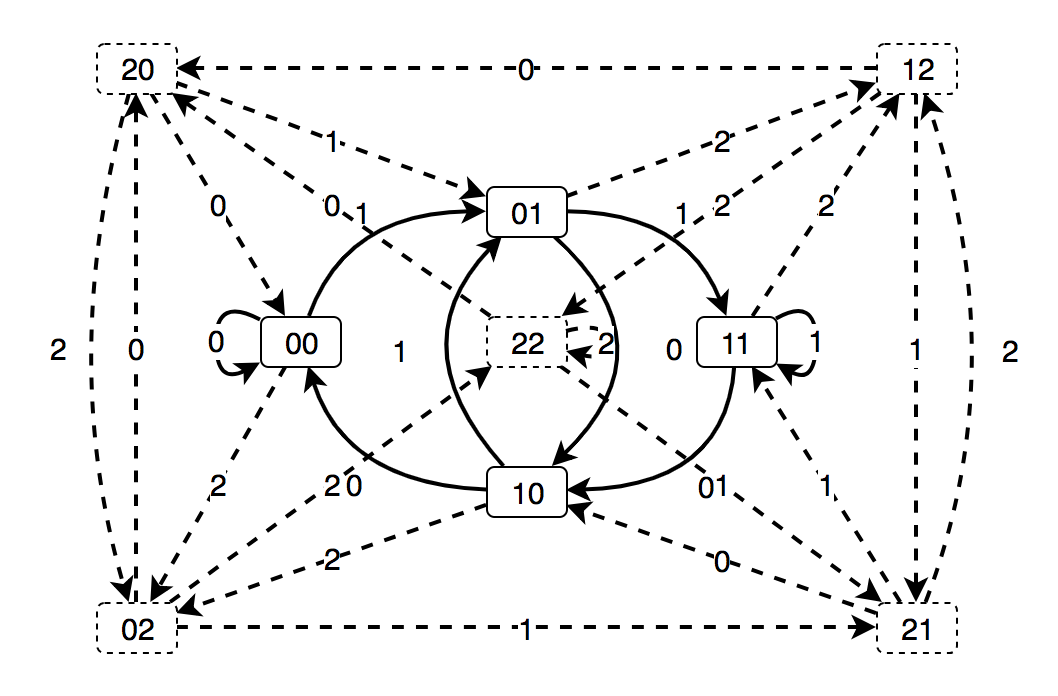}
  \caption{The de Bruijn  graph $G(2,2)$   is given by the solid  lines.
The Augmenting graph $A(3,2)$  consists of
all the vertices and just the dashed lines.}
 \label{fig:A}\end{figure}

Our main tool  is the  factorization of the set of edges in $G(k, n)$ 
in  convenient sets of pairwise disjoint  cycles. 
We say that two   cycles are  disjoint  if they have  no common edges.

\begin{Proposition}
For every, $k$ and $n$, the set of edges in $G(k, n)$ can be partitioned into a disjoint 
set of cycles identified by the circular words of length~$n+1$.
\end{Proposition}
\begin{proof}
As usual, we identify an edge  in $G(k,n)$ by 
concatenating the starting vertex label with  the edge  label.
Thus, each edge in $G(k,n)$ is identified with  a word of length~$n+1$.
The set of all rotations of a word of length $n+1$ identifies  consecutive edges that form a simple cycle in $G(k,n)$.
And each circular word of length  $n+1$  corresponds  exactly  to one simple cycle in $G(k,n)$. 
The  partition of the set of  words of length $n+1$ in the equivalence classes given by their  rotations 
determines a partition of the set of edges in $G(k, n)$ into disjoint simple cycles, see Figure~\ref{fig:simple}.
\end{proof}
We  define  the  graph of circular words. Figure~\ref{fig:C} shows it for word length $3$ over~$\{0,1\}$.

\begin{figure}[t!]
 \centering
  \includegraphics[width=.4\textwidth]{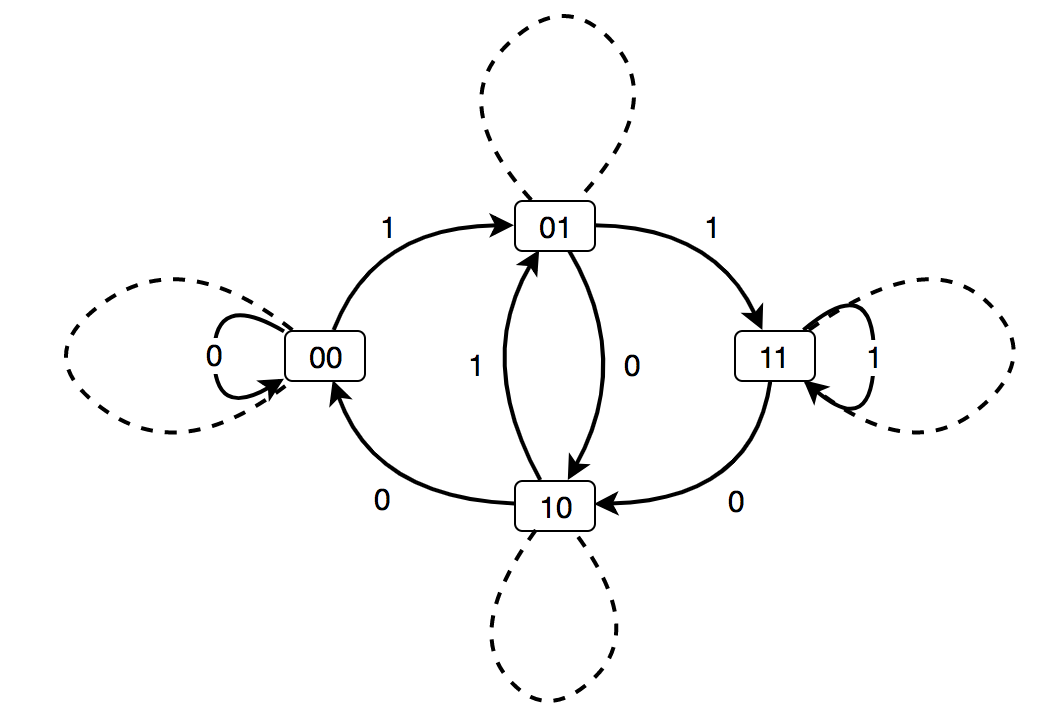}
  \caption{Petals for the vertices in $G(2,2)$. }
   \label{fig:petal}
  \includegraphics[width=\textwidth]{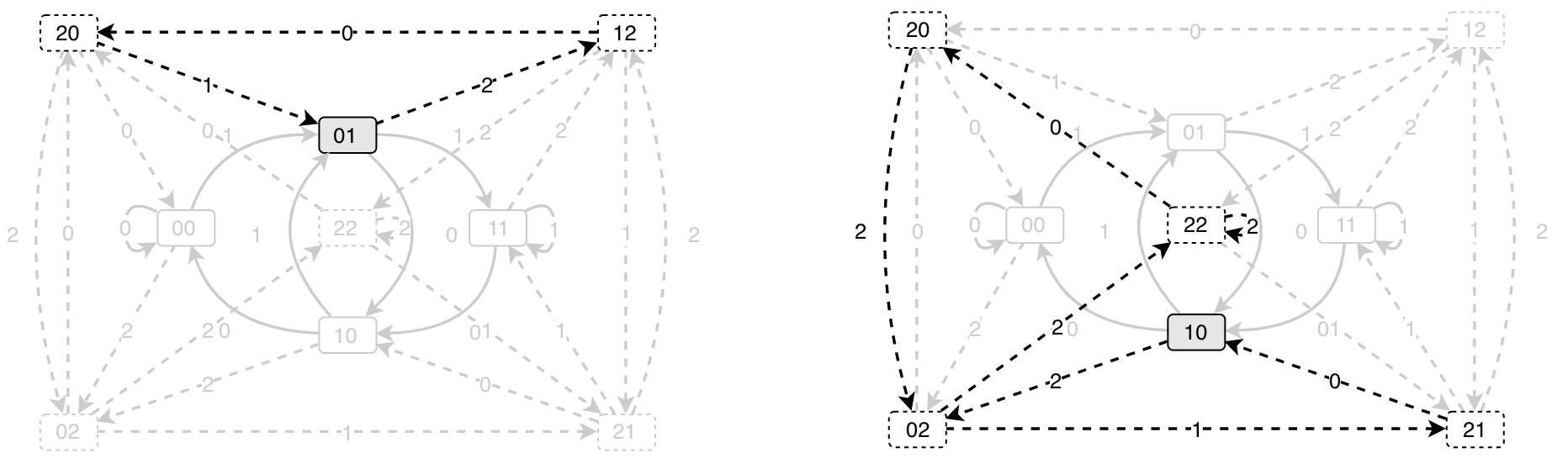}
  \caption{On the left,  the petal  for the vertex $01$, which  is just $[012]$. 
On the  right,  the petal  for the vertex $10$ which consists of the path  $[222]$,  $[202]$, $[021]$.}
  \label{fig:flower}
\end{figure}

\begin{Definition}[Graph of circular words]
For every  $k$ and $n$,   $C(k, n)$  is the graph whose  vertices are 
 the  circular words of length $n$ over the $k$-symbol  alphabet 
and two vertices  $[v]$ and $[w]$ are connected 
 if  there is a word $u$ of length $n-1$ and  symbols $a,b$ such that 
$[au]= [v]$,   $[ub]= [w]$.
\end{Definition}

The  fact that  $G(k,n)$ is a subgraph of $G(k+1,n)$ motivates the following definition.

\begin{Definition}[Augmenting graph]
The augmenting graph $ A(k+1, n)$ is the directed graph $(V, E)$ where 
$V$  is  the set of  length-$n$ words over the  alphabet enlarged by a new symbol $s$, 
and  $E$ is the set of pairs $(v,w)$  such that $v=au$, $w=ub$
for some word  $u$ of  length $n-1$  and  symbols $a,b$, 
and  either $v$ or $w$ have at least one occurrence of the symbol~$s$.
\end{Definition}

Figure \ref{fig:A} illustrates $A(3,2)$. Observe that in  $A(k+1,n)$ each of the vertices in $G(k,n)$ 
has exactly one incoming edge and exactly one outgoing edge.
This outcoming edge  is always labelled with the new symbol~$s$. 
To prove Theorem \ref{thm:1} we plan to construct an  Eulerian cycle in $G(k+1,n)$  by joining 
the given  Eulerian cycle in $G(k,n)$ with disjoint cycles of the augmenting graph $A(k+1,n)$ that we call {\em petals}.
Since the edges in $A(k+1,n)$ are exactly the edges in $G(k+1,n)$ minus those in $G(k,n)$, 
the edges in $A(k+1,n)$ can also be partitioned into a disjoint set of cycles which are identified
by  the circular words of length~$n+1$ that have at least one occurrence of the new symbol~$s$.
To define petals  we consider the  restriction of $C(k+1,n+1)$ to the simple cycles in~$A(k+1,n)$.

\begin{Definition}[Petal for a vertex in $G(k,n)$]
Let $\widetilde C(k+1, n+1)$  be the subgraph of $C(k+1, n+1)$ whose  set of vertices 
  are the circular words of length $n+1$ 
with at least one occurrence of symbol~$s$.
A {\em  petal} for a vertex $v$ in $G(k,n)$ is a subgraph of $\widetilde C(k+1,n+1)$ that seen as a cycle in $ A(k+1,n)$, 
 traverses  exactly one  vertex  in  $G(k, n)$, the vertex $v$.
\end{Definition}

There is exactly one  petal for each vertex $v$ in $G(k,n)$  and this petal  starts at the circular word~$[vs]$, 
where~$s$ is the new symbol.
Now there are two difficulties. 
One is to determine where to insert the petals so that we obtain a fair distribution of the new symbol~$s$.
The other difficulty is that petals must exhaust the augmenting graph~$A(k+1,n)$.
Figures~\ref{fig:petal} and \ref{fig:flower} illustrates petals for vertices in~$G(2,2)$.

\subsection{Fair distribution of the new symbol}

\begin{figure}[t!]
  \centering
  \includegraphics[width=0.35\textwidth]{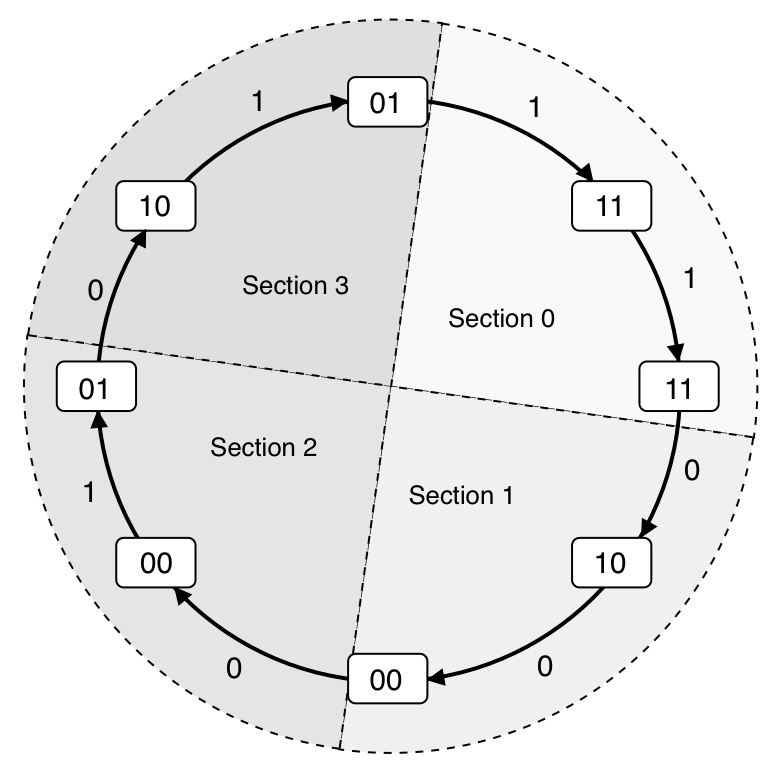}
  \caption{The Eulerian cycle in $G(2,2)$  given by $[11000101]$ started at  vertex $11$ 
 has 4 sections,  section $0$ is  $(11,11)$ , section 
$1$ is $(10,00)$,  section $2$ is $(00,01)$
and  section $3$ is $(10, 01)$.}
  \label{fig:sections}
\begin{tabular}{lr}
 \hspace{-.7cm} \includegraphics[width=0.5\textwidth]{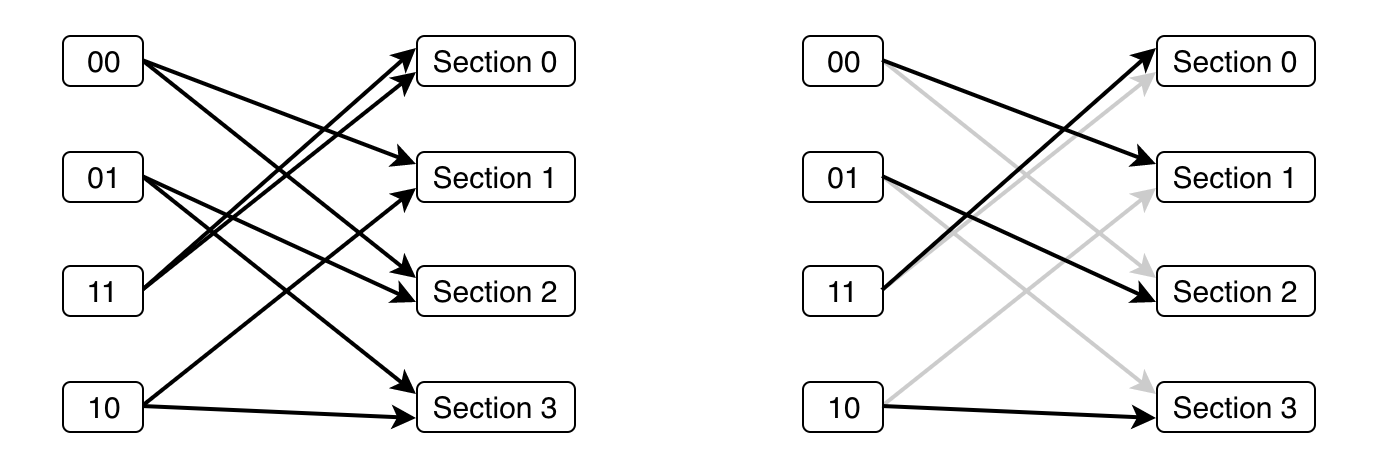}
&
 \hspace{-.5cm}   \includegraphics[width=0.4\textwidth]{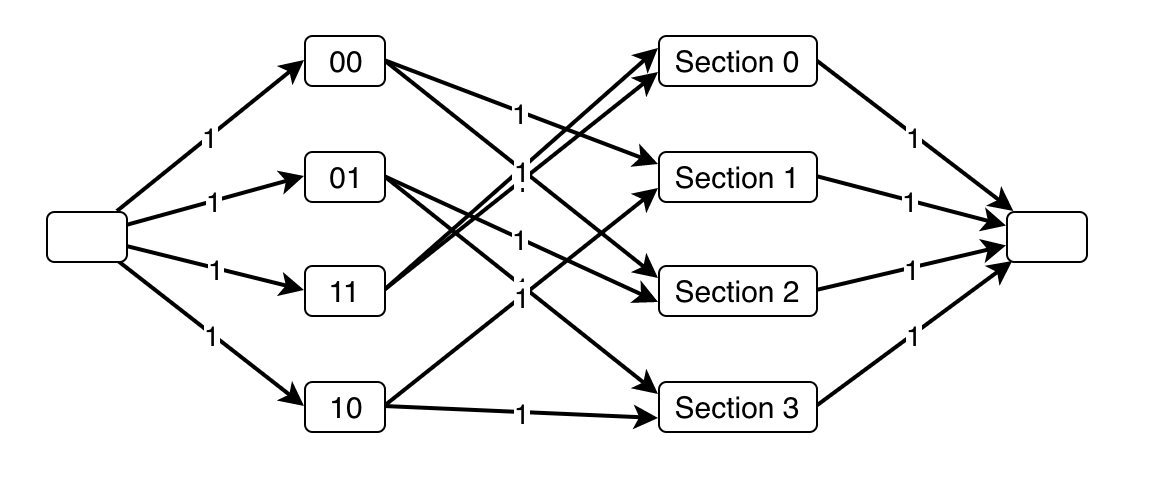}
\end{tabular}
 \caption{At the left, a  Distribution graph $D(2,2)$.
At the center, a possible perfect matching.
At the right,  the flow network for $D(2,2)$ where each edge has capacity $1$.}
\label{fig:D}
\end{figure}

A pointed cycle is a cycle with a specified starting edge.

\begin{Definition}[Section of a cycle]
For a pointed  Eulerian cycle in $G(k, n)$  given by the  sequence of  edges   $e_1, \ldots e_{k^{n+1}}$
and  a non-negative  integer $j$  such that $0\leq j< k^n$,
 the sequence  of vertices $v_{jk}, ..., v_{jk+k-1}$, where each $v_i$ is the head of $e_i$,
is  a section $j$ of the cycle.
\end{Definition}

Figure~\ref{fig:sections} exemplifies the four sections of an Eulerian cycle in~$G(2,2)$.
The de Bruijn  graph $G(k, n)$  has~$k^n$ vertices and~$k^{n+1}$ edges.
An Eulerian  cycle in  $G(k, n)$ has~$k^n$ sections with~$k$ vertices each section. 
Since  there are the  same number of vertices as sections 
we would like to choose one vertex from each section to place a petal. 
The problem is  each vertex  occurs  $k$ times in the Eulerian cycle  but not necessarily at $k$ different sections.
We pose it as a {\em matching} problem.

\begin{Definition}[Distribution graph]
Given pointed   Eulerian cycle  in  $G(k, n)$  
the {\em  Distribution graph} $D(k, n)$ is a $k$-regular bipartite graph 
where the two vertex classes
are the vertices in  $G(k, n)$ and the sections of  the Eulerian cycle
and  there is an  edge  $(v, j)$ if  $v$ belongs to the section~$j$.
\end{Definition}

A  matching  in  a graph $D$ is a set of edges such that no two edges share a common vertex. \linebreak
A vertex is  matched if it is an endpoint of one of the edges in the matching.
 A {\em perfect matching} is a matching that matches all vertices in the graph.

\begin{Lemma}
For every  Distribution graph $D(k,n)$ there is a perfect matching.
\end{Lemma}

\begin{proof}
Let $D$ be a finite bipartite graph consisting of 
 are  two disjoint  sets  of vertices $X$ and $Y$ with
 edges that  connect a vertex in $X$ to a vertex in $Y$.
For a subset $W$ of  $X$, let $N(W)$ be the  set of all vertices in $Y$ adjacent to some element in~$W$. 
Hall's marriage theorem~\cite{hallph} states that there is a matching that entirely 
covers $X$ if and only if for every subset $W$ in $X$, $|W| \leq |N(W)|$.
Consider a Distribution graph $D(k,n)$ and call   $X$ to   the set of vertices $G(k,n)$ and  $Y$ to  the set of sections.
For any  $W\subseteq X$ such that $|W| = r$, the sum of the out-degree of these $r$ vertices is $rk$. 
Given that the in-degree for any vertex in $Y$ is~$k$, we have that $|N(W)| \geq r$. 
Then, there is a matching that entirely covers $X$. 
Furthermore,  since the number of vertices is equal to the number of sections, 
 $|X| = |Y|$ and  the matching is perfect.
\end{proof}

To obtain a perfect matching in a Distribution 
graph we can use any method to compute the maximum flow in a network. 
We define the flow  network by adding
adding  two vertices  to the  Distribution graph,  the source and the sink.
Add an edge from the source  to each vertex in $X$ 
and add an edge from each vertex in $Y$ to the sink. 
Assign capacity~$1$ to each of the edges of the flow network. 
The maximum flow of the network is $|X|$.
This flow has the edges of a perfect match.
Figure \ref{fig:D} shows a Distribution graph $D(2,2)$, 
a possible perfect matching, and the flow network used to obtain it.

\subsection{Partition of the augmenting graph}
\begin{figure}
  \centering
  \includegraphics[width=0.4\textwidth]{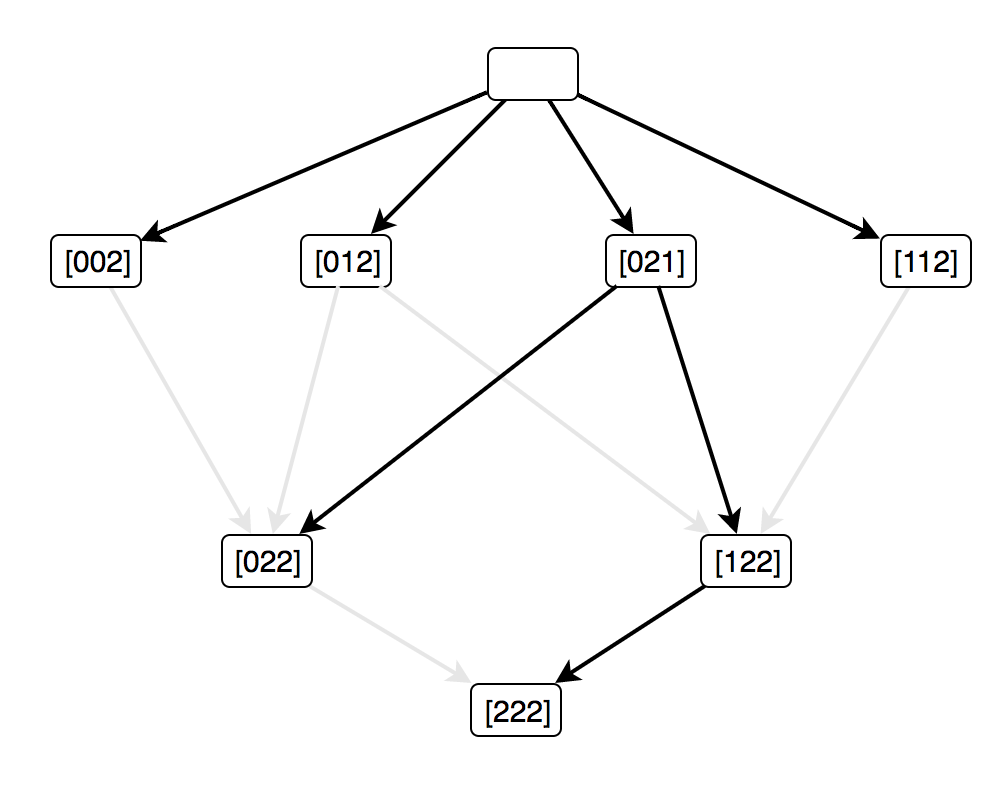}
  \caption{A  petals tree with four petals, one for each vertex of $G(2, 2)$.}
  \label{fig:tree}
\end{figure}

We must partition the set  of  edges in $ A(k+1,n)$ into petals.
We define a {\em Petals tree} as a  root  that branches out in a  subgraph of~$\widetilde C(k+1,n+1)$.
It   has height $n+1$,  the vertices at  distance~$d$ to the root have exactly $d$ occurrences 
of the new symbol~$s$, for $d=1, \ldots, n+1$. 

\begin{Definition}[Petals tree]
Let $[r]$ be a circular word corresponding to an Eulerian cycle in~$G(k,n)$.
We define the  {\em Petals tree }  given by the    root $[r]$ and   all the vertices in~$\widetilde C(k+1,n+1)$.
Every vertex~$[v]$ where $v$ has  exactly one occurrence of the symbol $s$ is a child of  the root~$[r]$.
And for every pair of vertices $[v]$, $[w]$ there is an edge between them 
 exactly when there is an edge between them  in~$\widetilde C(k+1,n+1)$
 and $w$ has one more occurrence  of the new symbol $s$ than~$v$.
\end{Definition}

Figure \ref{fig:tree} shows a petals tree.
 The root branches our in the  petal for vertex $00$, which  has the circular word $[002]$;
the petal for vertex $01$, which has  $[012]$;
the petal for vertex $10$, which  has $[021]$, $[022]$, $[122]$ ,$[222]$; and the petal for $11$ which has $[112]$.

\begin{figure}[t!]
  \centering
  \includegraphics[width=.85\textwidth]{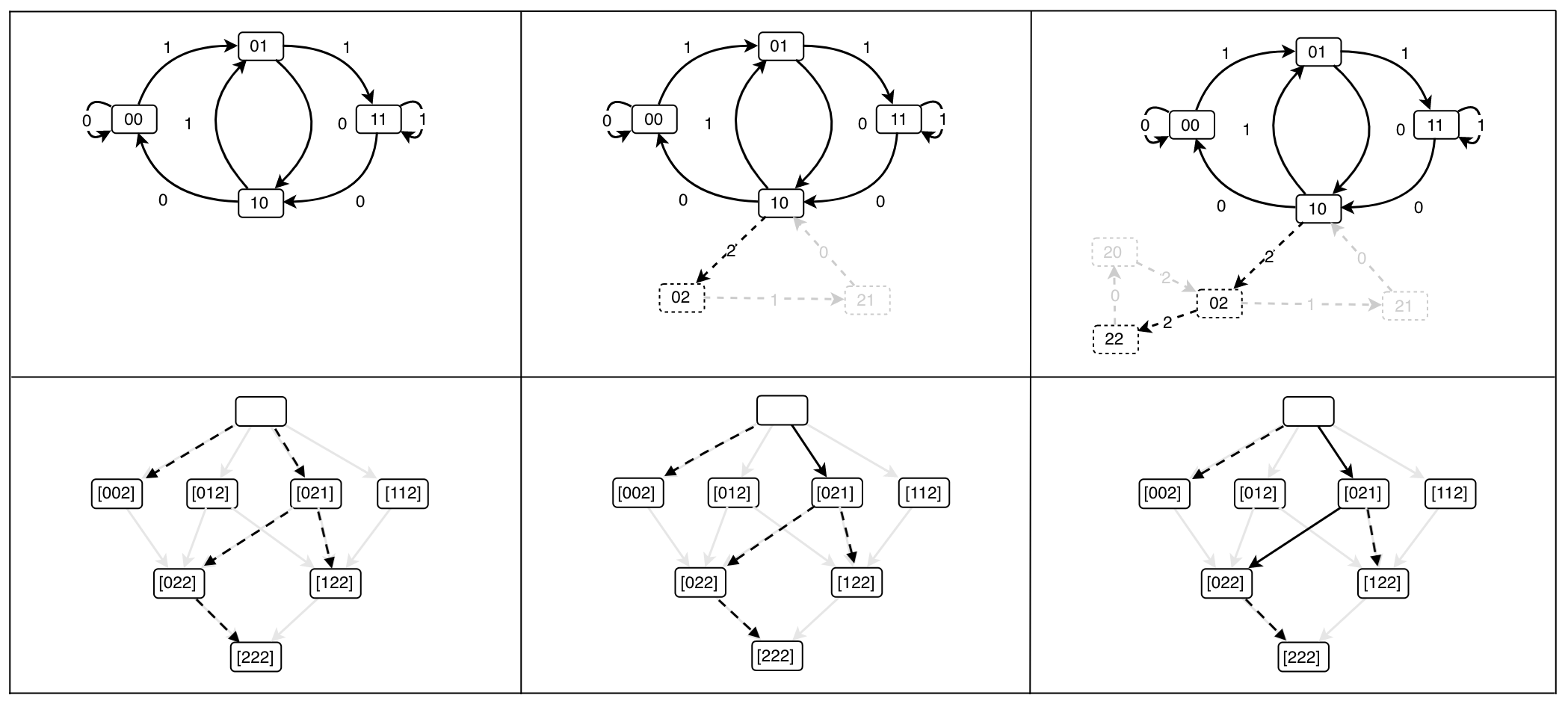}
\\\vspace{-0.1cm}
  \includegraphics[width=.85\textwidth]{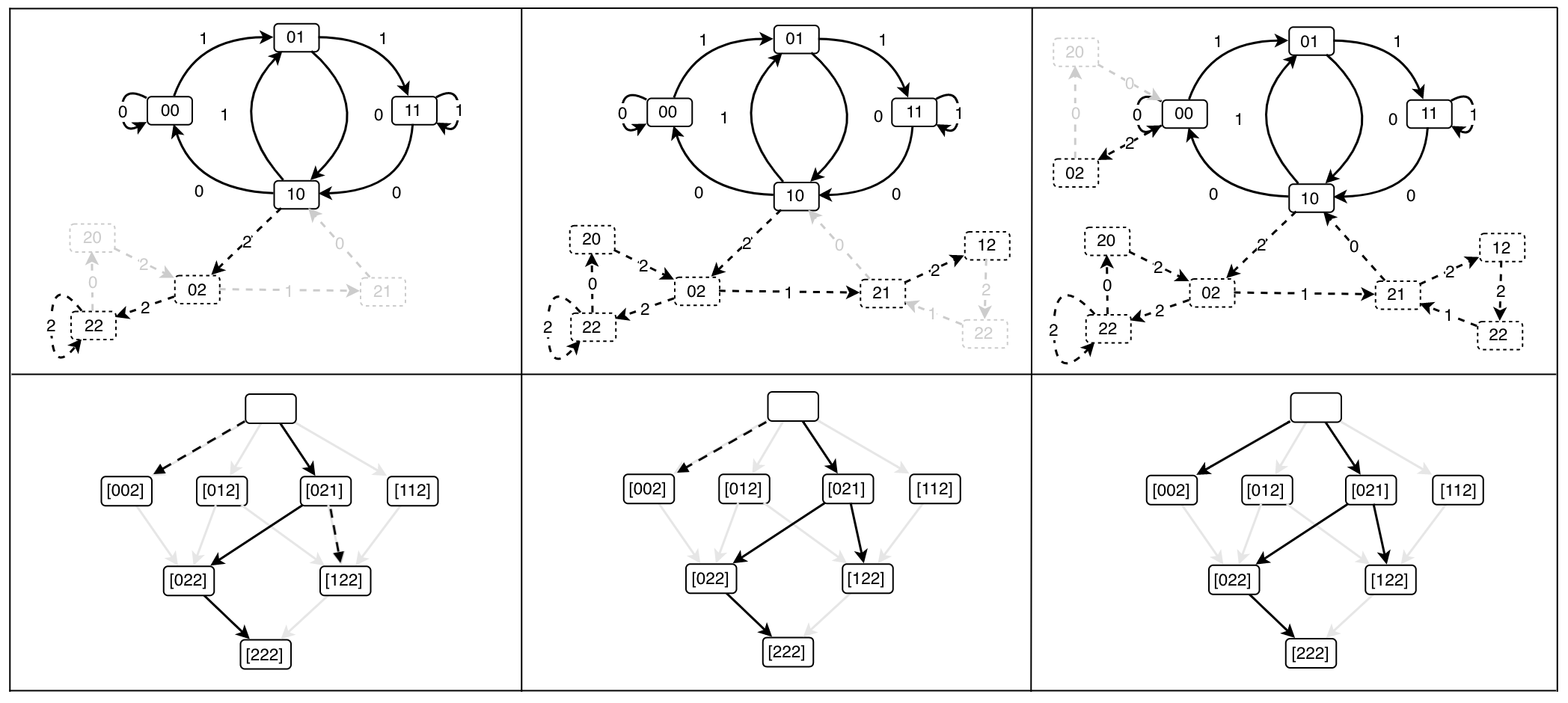}
\caption{Some steps of the construction of the extension of the de Bruiijn sequence $[11000101]$ }
  \label{fig:construction}
\end{figure}

Given  Eulerian cycle in $G(k,n)$ and a starting vertex, 
divide it in $k^n$ sections.
Choose one vertex in each section according to a perfect matching.
Fix a Petals tree as a subgraph of $A(k+1, n)$. 
The construction considers all the sections, one after the other, starting at section $0$.
At  each section the construction  inserts the  petal for a chosen  vertex, guided by the Petals tree.
Each traversed edge is added to the construction.
The construction starts at the vertex that is  the head of the first edge of section $0$.
Let $w$ be the current vertex.

{\em Case $w$ is a vertex in $G(k,n)$}:
If~$w$ is a chosen vertex in the current section and 
 the petal for $w$ has not been inserted yet then
traverse the edge labelled with symbol~$s$
and continue traversing  the petal for $w$ (which starts with  $[ws]$). 
If the petal for $w$ has already been traversed 
or $w$ is not a chosen vertex
then  continue with the traversal of edges in the current section.

{\em Case  $w$  is not a vertex in $G(k,n)$}:
If the  edge labelled with~$s$ has not been traversed yet,  $[ws]$ is a child  of the current node  the tree and 
 $[ws]$  has not been traversed yet, then traverse it. Otherwise continue with the traversal of the petal that 
$w$ was already part of.

For example,  consider this $[11000101]$  de Bruijn sequence of order~$3$ 
over  alphabet $\{0,1\}$  and the corresponding  Eulerian cycle in $G(2,2)$.
Suppose we start this cycle  at vertex~$11$ and 
 consider the four consecutive sections $(11, 10), (00, 00), (01,10) $ and $(01,11)$.
Assume   a perfect matching yields for section~$0$  the vertex~$10$ and 
for section~1 the second instance of the vertex~$ 00$.
Figure \ref{fig:construction} illustrates part the construction of the extended Eulerian cycle
for a given Petals tress, inserting the petal for the vertex~$10$ and the petal for the vertex~$00$.

\subsection{Actual proof of Theorem \ref{thm:1}}

\begin{proof}[Proof of Theorem \ref{thm:1}]
Let $e_1, ..., e_{k^n}$  be the list of edges visited by the Eulerian cycle determined by $v$ in $G(k,n-1)$,
and let $v_1, ..., v_{k^n}$ be the list of the respective head vertices.
Divide these vertices in $k^{n-1}$ consecutive sections,  each   has $k$ vertices. 
We use the Edmonds-Karp algorithm determine a vertex from each section. 
Consider the petals for $G(k, n-1)$. 
If we place one petal in each section,
two consecutive  petals can be at most $2k-1$  edges away. 
Consider now the  petals as pointed cycles in $ A(k+1, n-1)$.
A petal for  vertex $v=a_1 ... a_{n-1}$ 
starts with the outgoing edge labelled  $s$. 
Inside the petal,  for any $n-1$ consecutive edges there is one edge labelled with $s$.
The last edge of the petal is $(p,q)$ where 
$p=sa_1 ... a_{n-2}$ and $q=a_1 ... a_{n-1}$.
Thus, for any  $2k -1 + n$ consecutive edges there is at least  be one labelled with $s$.
\end{proof}

\subsection{Proof  of Proposition \ref{prop:1}}

\begin{proof}[Proof of Proposition \ref{prop:1}]
We must consider the Eulerian cycle in $G(k,n-1)$ and 
the extended Eulerian cycle in $G(k+1,n-1)$. 
The search of the maximum flow is the most expensive part of the construction.
Edmonds-Karp algorithm has running time $O(|V|^2 |E|)$, see~\cite{ek,Cormen}
for the flow graph $(V,E)$.
 In our case 
 $V$ has  a source, a sink, 
 $k^{n-1}$ vertices of the original de Bruijn graph $G(k,n-1)$
 and $k^{n-1}$  vertices for the sections. So $|V| = 2k^{n-1} +2$. 
There is an edge from the source  to each vertex in $G(k,n)$, 
 there are $k$ outgoing edges from each vertex in $G(k,n)$ to sections,
and there is one  outgoing edge to each section  the sink. 
So, $|E| =(k+2) \ k^{n-1}$. 
Then the time complexity of the  Edmonds-Karp algorithm in our graph is 
\[
O((2k^{n-1} +2 )^2  (k+2) k^{n-1}) = O(k^{3n-2}).
\] 
This completes the proof.
\end{proof}

\bibliographystyle{plain}
\bibliography{db3}

\end{document}